\documentclass{llncs}
\usepackage[english]{babel}
\usepackage[utf8]{inputenc}
\usepackage{amsmath,amssymb}
\usepackage{xspace}
\usepackage[inline]{enumitem}
\usepackage{multirow}
\usepackage[labelformat=simple]{subcaption}
\usepackage{csvsimple}
\usepackage{ifthen,xifthen}
\usepackage{wrapfig}
\usepackage{tabularx}
\usepackage{caption}
\usepackage{tikz}\usetikzlibrary{positioning,shapes,patterns,fit,external}
\tikzset{%
state/.style={draw,circle,thick},%
final/.style={double},%
transition/.style={->,>=stealth,thick},%
inmeta/.style={state,dashed},%
outmeta/.style={state,rectangle,rounded corners},%
hist/.style={rectangle,rounded corners,inner sep=0cm,fill=black!10},%
s/.style={state},%
t/.style={transition}%
}
\usepackage{pgfplots,pgfplotstable} \pgfplotsset{compat=1.9,scaled ticks=false,mylegend/.style={legend columns=1,legend cell align=left},myscale/.style={scaled ticks=true,x tick scale label style={at={(xticklabel cs:1.12,-8mm)}}}} \pgfplotstableset{col sep=comma} \usetikzlibrary{pgfplots.statistics}
\usepackage{hyperref}
\usepackage[all]{hypcap} 
\hypersetup{%
colorlinks=true,
linkcolor=blue,
citecolor=red,
urlcolor=blue
}
\newenvironment{enumin}[3][.]{\begin{enumerate*}[label=\arabic*),itemjoin={{#2 }},itemjoin*={{#2 #3 }},after={#1}]}{\end{enumerate*}}
\setlist{itemsep=2pt}
\setlist[2]{topsep=2pt}
\newif\ifextendedversion
\newif\iftacasversion
\extendedversiontrue
\newcommand{\vpa}{\textsc{Vpa}\xspace} 
\newcommand{\ba}{\textsc{Ba}\xspace} 
\newcommand{\fa}{\textsc{Fa}\xspace} 
\newcommand{\sevpa}{\textsc{Sevpa}\xspace} 
\newcommand{\vpl}{\textsc{Vpl}\xspace} 
\newcommand{\sat}{\textsc{Sat}\xspace} 
\newcommand{\pmaxsat}{\textsc{PM}ax-\textsc{Sat}\xspace} 
\newcommand{\np}{\textsc{Np}\xspace} 
\newcommand{\resp}{resp.\xspace} 
\newcommand{\Qrel}{\textsc{Raq} relation\xspace} 
\newcommand{\ultimate}{\textsc{Ultimate}\xspace} 
\newcommand{\automizer}{\textsc{Automizer}\xspace} 
\newcommand{\cegar}{CEGAR\xspace} 
\newcommand{\svcomp}{SV-COMP\xspace} 
\newcommand{\AlphaAll}{\ensuremath{\Sigma}\xspace} 
\newcommand{\AlphaI}{\ensuremath{\Sigma_\text{i}}\xspace} 
\newcommand{\AlphaC}{\ensuremath{\Sigma_\text{c}}\xspace} 
\newcommand{\AlphaR}{\ensuremath{\Sigma_\text{r}}\xspace} 
\newcommand{\AlphaS}{\ensuremath{Q_{\bot}}\xspace} 
\newcommand{\transAll}{\ensuremath{\Delta}\xspace} 
\newcommand{\transI}{\ensuremath{\Delta_\text{i}}\xspace} 
\newcommand{\transC}{\ensuremath{\Delta_\text{c}}\xspace} 
\newcommand{\transR}{\ensuremath{\Delta_\text{r}}\xspace} 
\newcommand{\eps}{\ensuremath{\varepsilon}\xspace} 
\newcommand{\wellmatched}{\ensuremath{\mathit{WM}(\AlphaAll)}\xspace}
\newcommand{\matchedreturn}{\ensuremath{\mathit{MR}(\AlphaAll)}\xspace}
\newcommand{\run}[1][]{\ensuremath{\rho}\xspace} 
\newcommand{\stacks}{\ensuremath{\mathit{St}}\xspace} 
\newcommand{\stackRAW}{\sigma} 
\newcommand{\stack}{\ensuremath{\stackRAW}\xspace} 
\newcommand{\A}{\ensuremath{\mathcal{A}}\xspace} 
\newcommand{\B}{\ensuremath{\mathcal{B}}\xspace} 
\newcommand{\lang}{\ensuremath{L}\xspace} 
\newcommand{\Top}{\ensuremath{\mathit{top}}\xspace} 
\newcommand{\tops}{\ensuremath{\mathit{tops}}\xspace} 
\newcommand{\hier}[1]{\ensuremath{\hat{#1}}\xspace} 
\newcommand{\symb}{\ensuremath{x}\xspace} 
\newcommand{\true}{\ensuremath{\textsf{true}}\xspace} 
\newcommand{\false}{\ensuremath{\textsf{false}}\xspace} 
\newcommand{\XRAW}{\mathsf{X}} 
\newcommand{\X}[1]{\ensuremath{\XRAW_{\{#1\}}}\xspace} 
\newcommand{\Xa}[1]{\ensuremath{\XRAW_{#1}}\xspace} 
\newcommand{\cnf}{\ensuremath{\Phi}\xspace} 
\newcommand{\etal}{\textit{et al.}\xspace}
\newcommand{\defeq}{\ensuremath{\mathop{\stackrel{\text{def}}{=}}}\xspace} 
\newcommand{\eqrel}{\ensuremath{\equiv}\xspace} 
\newcommand{\concat}{\ensuremath{\cdot}} 
\newcommand{\quantDot}{\ensuremath{.\,}} 
\newcommand{\initialtransition}[2][5mm]{\draw[t,<-] (#2) -- +(-#1,#1);} 
\newcommand{\On}[1]{\ensuremath{\mathcal{O}(#1)}\xspace} 
\newcommand{\colhead}[1]{\multicolumn{1}{c|}{#1}}
\newcommand{\colheadDouble}[1]{\multicolumn{1}{c||}{#1}}
\newcommand{\myitem}{\item[$\bullet$]}
\newcommand{\mysubsubsection}[1]{\vspace*{-\baselineskip}\subsubsection{#1}}
\makeatletter
\newcommand{\blfootnote}{\xdef\@thefnmark{}\@footnotetext}
\makeatother
\colorlet{darkgreen}{green!80!black}
\newcommand{\plot}[5]{%
\begin{tikzpicture}
	\begin{axis}[mylegend,xlabel={#1},width=#3,height=#3,ylabel={#2},#4]
		#5
	\end{axis}
\end{tikzpicture}
}
\newcommand{\scatterPlot}[9][]{%
	\plot{#2}{#3}{45mm}{myscale,domain=1:#7,xmin=1,xmax=#7,xtick=#8,ymin=1,ymax=#7,ytick=#8,#1}{
		\addplot[mark=none,color=red,thick] {x};
		\addplot[mark=none,color=red,thick,dashed] {0.5*x};
		\addplot[only marks,mark=o,mark size=.7,color=blue] table [col sep=comma,x=#5,y=#6] {data/#4.csv};
		\ifthenelse{\isempty{#9}}{}{
			\addplot[only marks,mark=asterisk,mark size=1.1,color=darkgreen] table [col sep=comma,x=#5,y=#6] {data/#9.csv};
		}
	}
}
\newcommand{\dotPlotWithMeanHardCodedStates}[8][]{%
	\plot{#2}{#3}{45mm}{myscale,#1}{%
		\pgfmathsetmacro{\MeanA}{38.14} \DrawHMean[blue,thick]{A}
		\ifthenelse{\isempty{#8}}{}{
			\pgfmathsetmacro{\MeanB}{36.54} \DrawHMean[darkgreen,thick]{B}
		}
		#7
		\addplot+[only marks,mark=o,mark size=.7,color=blue] table [col sep=comma,x=#5,y=#6] {data/#4.csv};
		\ifthenelse{\isempty{#8}}{}{
			\addplot+[only marks,mark=asterisk,mark size=1.1,color=darkgreen] table [col sep=comma,x=#5,y=#6] {data/#8.csv};
		}
	}
}
\newcommand{\dotPlotWithMeanHardCodedTransitions}[8][]{%
	\plot{#2}{#3}{45mm}{myscale,#1}{%
		\pgfmathsetmacro{\MeanA}{36.53} \DrawHMean[blue,thick]{A}
		\ifthenelse{\isempty{#8}}{}{
			\pgfmathsetmacro{\MeanB}{35.89} \DrawHMean[darkgreen,thick]{B}
		}
		#7
		\addplot+[only marks,mark=o,mark size=.7,color=blue] table [col sep=comma,x=#5,y=#6] {data/#4.csv};
		\ifthenelse{\isempty{#8}}{}{
			\addplot+[only marks,mark=asterisk,mark size=1.1,color=darkgreen] table [col sep=comma,x=#5,y=#6] {data/#8.csv};
		}
	}
}
\newcommand{\dotPlot}[8][]{%
	\plot{#2}{#3}{45mm}{myscale,#1}{%
		\GetMeanA{data/#4Aggregated.csv}{#6} \DrawHMean[blue,thick]{A}
		\ifthenelse{\isempty{#8}}{}{
			\GetMeanB{data/#8Aggregated.csv}{#6} \DrawHMean[darkgreen,thick]{B}
		}
		#7
		\addplot+[only marks,mark=o,mark size=.7,color=blue] table [col sep=comma,x=#5,y=#6] {data/#4.csv};
		\ifthenelse{\isempty{#8}}{}{
			\addplot+[only marks,mark=asterisk,mark size=1.1,color=darkgreen] table [col sep=comma,x=#5,y=#6] {data/#8.csv};
		}
	}
}
\newcommand\GetMeanA[2]{
	\pgfplotstableread{#1}\tableA
	\pgfplotstableset{
		create on use/new/.style={
		create col/expr={\pgfmathaccuma + \thisrow{#2}}},
	}
	\pgfplotstablegetrowsof{\tableA}
	\pgfmathsetmacro{\NumRows}{\pgfplotsretval}
	\pgfplotstablegetelem{\numexpr\NumRows-1\relax}{new}\of{#1} 
	\pgfmathsetmacro{\Sum}{\pgfplotsretval}
	\pgfmathsetmacro{\MeanA}{\Sum/\NumRows}
}
\newcommand\GetMeanB[2]{
	\pgfplotstableread{#1}\tableA
	\pgfplotstableset{
		create on use/new/.style={
		create col/expr={\pgfmathaccuma + \thisrow{#2}}},
	}
	\pgfplotstablegetrowsof{\tableA}
	\pgfmathsetmacro{\NumRows}{\pgfplotsretval}
	\pgfplotstablegetelem{\numexpr\NumRows-1\relax}{new}\of{#1} 
	\pgfmathsetmacro{\Sum}{\pgfplotsretval}
	\pgfmathsetmacro{\MeanB}{\Sum/\NumRows}
}
\newcommand\DrawVMean[2][]{
\draw[#1] 
  (axis cs:\csname Mean#2\endcsname,\pgfkeysvalueof{/pgfplots/ymin}) -- 
  (axis cs:\csname Mean#2\endcsname,\pgfkeysvalueof{/pgfplots/ymax});
}
\newcommand\DrawHMean[2][]{
\draw[#1] 
  (axis cs:\pgfkeysvalueof{/pgfplots/xmin},\csname Mean#2\endcsname) -- 
  (axis cs:\pgfkeysvalueof{/pgfplots/xmax},\csname Mean#2\endcsname);
}
\newcommand{\doublelineplot}[8][]{%
	\plot{#2}{#3}{45mm}{myscale,#1}{%
		\addplot[no marks,thick,color=blue] table [col sep=comma,x=INTERNAL_DENSITY,y=#4] {data/#8.csv};
		#6
		\addplot[no marks,thick,color=red] table [col sep=comma,x=INTERNAL_DENSITY,y=#5] {data/#8.csv};
		#7
	}
}
\newcommand{\plottable}[6]{%
\csvreader[%
	separator=comma,tabular=#2,
	table head=\hline #3 \\ \hline,
	table foot=\hline,
	late after line=\\,
	#6
	]{data/#1.csv}%
	{#5}%
	{#4}%
}

\newcommand{\plottableOffline}[4]{%
	\plottable{#1}{#2}{#3}{#4}{%
		Aggregation=\KEYaggregation,
		Count=\KEYcount,
		STATES_INPUT=\KEYstatesIn,
		STATES_OUTPUT=\KEYstatesOut,
		STATES_REDUCTION_ABSOLUTE=\KEYstatesRemoved,
		STATES_REDUCTION_RELATIVE=\KEYstatesRemovedRel,
		SIZE_MAXIMAL_INITIAL_EQUIVALENCE_CLASS=\KEYmaxInitPartBlock,
		ALPHABET_SIZE_INTERNAL=\KEYsymbolsInternal,
		ALPHABET_SIZE_CALL=\KEYsymbolsCall,
		ALPHABET_SIZE_RETURN=\KEYsymbolsReturn,
		TRANSITIONS_INTERNAL_INPUT=\KEYtransitionsInternalIn,
		TRANSITIONS_CALL_INPUT=\KEYtransitionsCallIn,
		TRANSITIONS_RETURN_INPUT=\KEYtransitionsReturnIn,
		TRANSITIONS_INTERNAL_OUTPUT=\KEYtransitionsInternalOut,
		TRANSITIONS_CALL_OUTPUT=\KEYtransitionsCallOut,
		TRANSITIONS_RETURN_OUTPUT=\KEYtransitionsReturnOut,
		BUCHI_NONDETERMINISTIC_STATES=\KEYnondeterminism,
		RUNTIME_TOTAL=\KEYruntimeTotal,
		NUMBER_OF_VARIABLES=\KEYvariables,
		NUMBER_OF_CLAUSES=\KEYclauses
	}
}
\title{Minimization~of~Visibly~Pushdown~Automata Using~Partial~Max-SAT}
\author{Matthias Heizmann \and Christian Schilling \and Daniel Tischner}
\institute{University of Freiburg, Germany}
\begin{document}

\maketitle

\ifextendedversion{%
\blfootnote{This paper is an extended version of the paper with the same title presented at TACAS 2017~\cite{TACAS2017}.
The final publication is available at Springer via \url{http://dx.doi.org/10.1007/978-3-662-54577-5_27}.}%
}\fi

\begin{abstract}

We consider the problem of state-space reduction for nondeterministic weakly-hierarchical visibly pushdown automata (\vpa).
\vpa recognize a robust and algorithmically tractable fragment of context-free languages that is natural for modeling programs.

We define an equivalence relation that is sufficient for language-preserving quotienting of \vpa.
Our definition allows to merge states that have different behavior, as long as they show the same behavior for reachable equivalent stacks.
We encode the existence of such a relation as a Boolean partial maximum satisfiability (\pmaxsat) problem and present an algorithm that quickly finds satisfying assignments.
These assignments are sub-optimal solutions to the \pmaxsat problem but can still lead to a significant reduction of states.

We integrated our method in the automata-based software verifier \ultimate \automizer and show performance improvements on benchmarks from the software verification competition \svcomp.

\end{abstract}





\section{Introduction}\label{sec:introduction}

The class of visibly pushdown languages (\vpl)~\cite{AlurMadhusudan:2004} lies properly between the regular and the context-free languages.
\vpl enjoy most desirable properties of regular languages (closure under Boolean operations and decision procedures for, e.g., the equivalence problem).
They are well-suited for representing data that have both a linear and a hierarchical ordering, e.g., procedural programs~\cite{ThakurLLBDEAR10,HeizmannHP10,HarrisJR12,AlurBE17} and XML documents~\cite{Pitcher05,KumarMV07,ThomoVY08,MozafariZZ12}. 

The corresponding automaton model is called \emph{visibly pushdown automaton} (\vpa).
It extends the finite automaton model with a stack of restricted access by requiring that the input symbol specifies the stack action -- a call (\resp return) symbol implies a push (\resp pop) operation, and an internal symbol ignores the stack.
In this paper, we consider a notion of \vpa where a call always pushes the current state on the stack. These \vpa are called weakly-hierarchical \vpa~\cite{AlurMadhusudan:2009}.

\smallskip

Size reduction of automata is an active research topic~\cite{Clemente11,MayrClemente:2013,BaarirDuretlutz:2014,DAntoniV14,BaarirD15,AbelR15,AlmeidaEtal:2016} that is theoretically appealing and has practical relevance: smaller automata require less memory and speed up automata-based tools~\cite{HolzmannP99,KlarlundMS02,HabermehlHRSV12,HeizmannDGLMSP16}. 
In this paper, we present a size reduction technique for a general class of (nondeterministic) \vpa that is different from classes that were considered in previous approaches~\cite{AlurEtal:2005,KumarEtal:2006,ChervetWalukiewicz:2007}.
\iftacasversion{An extended version of this paper is available~\cite{extendedVersionArxiv}.}\fi

\smallskip

It is well-known that for deterministic finite automata the unique minimal automaton can be obtained by quotienting (i.e., merging equivalent states), and there exists an efficient algorithm for this purpose~\cite{Hopcroft:1971}.
\vpa do not have a canonical minimum~\cite{AlurEtal:2005}.
For other automaton classes that lack this property, the usual approach is to find equivalence relations that are sufficient for quotienting~\cite{EtessamiEtal:2005,AbdullaEtal:2009,AlmeidaEtal:2016}.
The main difficulty of a quotienting approach for \vpa is that two states may behave similarly given one stack but differently given another stack, and as the number of stacks is usually infinite, one cannot simply compare the behaviors for each of them.

\subsection{Motivating examples}

We now present three observations.
The first observation is our key insight and shows that \vpa have interesting properties that we can exploit.
The other observations show that \vpa have intricate properties that make quotienting nontrivial.
For convenience, we use $a$ for internal, $c$ for call, and $r$ for return symbols, and we omit transitions to the sink state.

\subsubsection{Exploiting unreachable stacks allows merging states}

\begin{figure}[t]
	\centering
	\begin{subfigure}[b]{0.35\textwidth}
		\centering
	\ifextendedversion{%
		\includegraphics{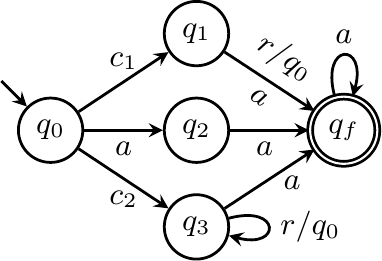}
	}\else{%
		\tikzsetnextfilename{auto_two_out_of_three_original}
\begin{tikzpicture}
	\node[state] (q0) {$q_0$};
	\node[state] (q2) [right=8mm of q0] {$q_2$};
	\node[state] (q1) [above=3mm of q2] {$q_1$};
	\node[state] (q3) [below=3mm of q2] {$q_3$};
	\node[state,final] (qf) [right=8mm of q2] {$q_f$};
	%
	\initialtransition[5mm]{q0}
	\draw[t] (q0) to node[above] {$c_1$} (q1);
	\draw[t] (q0) to node[below] {$a$} (q2);
	\draw[t] (q0) to node[below] {$c_2$} (q3);
	\draw[t,sloped] (q1) to node[above] {$r/q_0$} (qf);
	\draw[t,sloped,draw=none] (q1) to node[below] {$a$} (qf);
	\draw[t] (q2) to node[below,xshift=-.5mm] {$a$} (qf);
	\draw[t] (q3) to node[below,near end] {$a$} (qf);
	\draw[t,loop right] (q3) to node[right] {$r/q_0$} (q3);
	\draw[t,loop above] (qf) to node[above] {$a$} (qf);
\end{tikzpicture}
	}\fi

		\caption{A \vpa.}
		\label{fig:auto_two_out_of_three_1}
	\end{subfigure}
	\ifextendedversion{\\[-13mm]}\else{\\}\fi
	\begin{subfigure}[b]{0.4\textwidth}
		\centering
	\ifextendedversion{%
		\includegraphics{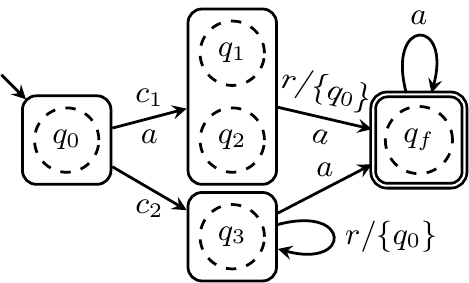}
	}\else{%
		\tikzsetnextfilename{auto_two_out_of_three_merged_1}
\begin{tikzpicture}
	\node[inmeta] (q0) {$q_0$};
	\node[outmeta,fit=(q0)] (meta0) {};
	\node[inmeta] (q2) [right=10mm of q0] {$q_2$};
	\node[inmeta] (q1) [above=2mm of q2] {$q_1$};
	\node[outmeta,fit=(q1)(q2)] (meta12) {};
	\node[inmeta] (q3) [below=3mm of q2] {$q_3$};
	\node[outmeta,fit=(q3)] (meta3) {};
	\node[inmeta] (qf) [right=12mm of q2] {$q_f$};
	\node[outmeta,fit=(qf),final] (metaf) {};
	\initialtransition[2mm]{meta0.north west) +(.5mm,-.5mm}
	\draw[t] (meta0) to node[above] {$c_1$} (meta12);
	\draw[t,draw=none] (meta0) to node[below] {$a$} (meta12);
	\draw[t] (meta0) to node[below] {$c_2$} (meta3);
	\draw[t,sloped] (meta12) to node[above,xshift=-0.5mm] {$r/\{q_0\}$} (metaf);
	\draw[t,sloped,draw=none] (meta12) to node[below] {$a$} (metaf);
	\draw[t] (meta3) to node[above] {$a$} (metaf);
	\draw[t,loop right] (meta3) to node[right] {$r/\{q_0\}$} (meta3);
	\draw[t,loop above] (metaf) to node[above] {$a$} (metaf);
\end{tikzpicture}
	}\fi

		\caption{One possible quotienting.}
		\label{fig:auto_two_out_of_three_2}
	\end{subfigure}
	\hfill
	\begin{subfigure}[b]{0.4\textwidth}
		\centering
	\ifextendedversion{%
		\includegraphics{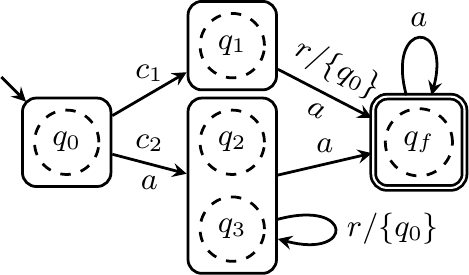}
	}\else{%
		\tikzsetnextfilename{auto_two_out_of_three_merged_2}
\begin{tikzpicture}
	\node[inmeta] (q0) {$q_0$};
	\node[outmeta,fit=(q0)] (meta0) {};
	\node[inmeta] (q2) [right=10mm of q0] {$q_2$};
	\node[inmeta] (q1) [above=3mm of q2] {$q_1$};
	\node[outmeta,fit=(q1)] (meta1) {};
	\node[inmeta] (q3) [below=2mm of q2] {$q_3$};
	\node[outmeta,fit=(q2)(q3)] (meta23) {};
	\node[inmeta] (qf) [right=12mm of q2] {$q_f$};
	\node[outmeta,fit=(qf),final] (metaf) {};
	\initialtransition[2mm]{meta0.north west) +(.5mm,-.5mm}
	\draw[t] (meta0) to node[above] {$c_1$} (meta1);
	\draw[t,draw=none] (meta0) to node[below] {$a$} (meta23);
	\draw[t] (meta0) to node[above] {$c_2$} (meta23);
	\draw[t,sloped] (meta1) to node[above] {$r/\{q_0\}$} (metaf);
	\draw[t,sloped,draw=none] (meta1) to node[below] {$a$} (metaf);
	\draw[t] (meta23) to node[above] {$a$} (metaf);
	\draw[t,loop right,distance=8mm] (meta23.east |- q3.east)+(0mm,1mm) to node[right] {$r/\{q_0\}$} +(0mm,-1mm);
	\draw[t,loop above] (metaf) to node[above] {$a$} (metaf);
\end{tikzpicture}
	}\fi

		\caption{Another possible quotienting.}
		\label{fig:auto_two_out_of_three_3}
	\end{subfigure}
	\caption{A \vpa and two possible quotientings due to unreachable stacks.}
	\label{fig:auto_two_out_of_three}
\end{figure}
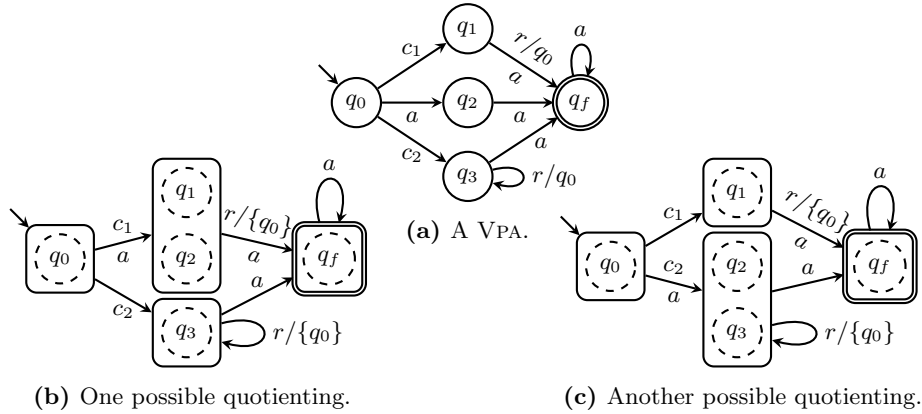

Consider the \vpa in Figure~\ref{fig:auto_two_out_of_three_1}.
The states $q_1$ and $q_2$ have the same behavior for the internal symbol $a$ but different behaviors for the return symbol $r$ with stack symbol $q_0$:
Namely, state $q_1$ leads to the accepting state while $q_2$ has no respective return transition.
However, in $q_2$ it is generally impossible to take a return transition with stack symbol $q_0$ since $q_2$ can only be reached with an empty stack.
Thus the behavior for the stack symbol $q_0$ is ``undefined'' and we can merge $q_1$ and $q_2$ without changing the language.
The resulting \vpa is depicted in Figure~\ref{fig:auto_two_out_of_three_2}.

\subsubsection{Merging states requires a transitive relation}
Using the same argument as above, we can also merge the states $q_2$ and $q_3$; the result is depicted in Figure~\ref{fig:auto_two_out_of_three_3}.
For finite automata, mergeability of states is transitive.
However, here we cannot merge all three states $q_1$, $q_2$, and $q_3$ without changing the language because $q_1$ and $q_3$ have different behaviors for stack symbol $q_0$.
For \vpa, we have to check compatibility for each pair of states.

\subsubsection{Merging states means merging stack symbols}

\begin{figure}[t]
	\begin{subfigure}[b]{0.4\textwidth}
		\centering
	\ifextendedversion{%
		\includegraphics{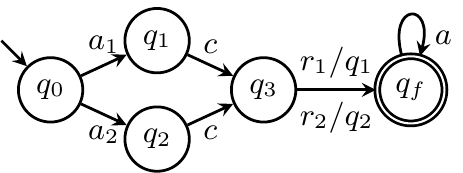}
	}\else{%
		\tikzsetnextfilename{auto_merge_stack_symbols_original}
\begin{tikzpicture}
	\node[state] (q0) {$q_0$};
	\node[state] (q1) [right=4mm of q0,yshift=5mm] {$q_1$};
	\node[state] (q2) [right=4mm of q0,yshift=-5mm] {$q_2$};
	\node[state] (q3) [right=4mm of q1,yshift=-5mm] {$q_3$};
	\node[state,final] (qf) [right=8mm of q3] {$q_f$};
	\initialtransition[5mm]{q0}
	\draw[t] (q0) to node[above] {$a_1$} (q1);
	\draw[t] (q0) to node[below] {$a_2$} (q2);
	\draw[t] (q1) to node[above] {$c$} (q3);
	\draw[t] (q2) to node[below] {$c$} (q3);
	\draw[t,sloped] (q3) to node[above] {$r_1/q_1$} node[below] {$r_2/q_2$} (qf);
	\draw[t,loop above] (qf) to node[right,very near end] {$a$} (qf);
\end{tikzpicture}
	}\fi

		\caption{A \vpa.}
		\label{fig:auto_merge_stack_symbols_1}
	\end{subfigure}
	\hfill
	\begin{subfigure}[b]{0.56\textwidth}
		\centering
	\ifextendedversion{%
		\includegraphics{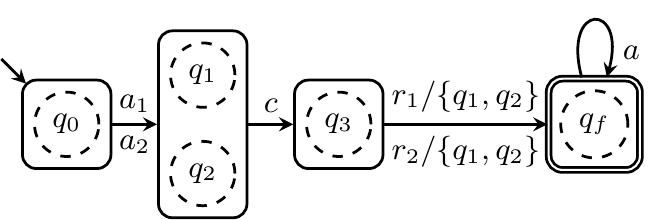}
	}\else{%
		\tikzsetnextfilename{auto_merge_stack_symbols_merged}
\begin{tikzpicture}
	\node[inmeta] (q0) {$q_0$};
	\node[outmeta,fit=(q0)] (meta0) {};
	\node[inmeta] (q1) [right=7mm of q0,yshift=5mm] {$q_1$};
	\node[inmeta] (q2) [right=7mm of q0,yshift=-5mm] {$q_2$};
	\node[outmeta,fit=(q1)(q2)] (meta12) {};
	\node[inmeta] (q3) [right=7mm of q1,yshift=-5mm] {$q_3$};
	\node[outmeta,fit=(q3)] (meta3) {};
	\node[inmeta] (qf) [right=19mm of q3] {$q_f$};
	\node[outmeta,final,fit=(qf)] (metaf) {};
	\initialtransition[2mm]{meta0.north west) +(.5mm,-.5mm}
	\draw[t] (meta0) to node[above] {$a_1$} node[below] {$a_2$} (meta12);
	\draw[t] (meta12) to node[above] {$c$} (meta3);
	\draw[t,sloped] (meta3) to node[above] {$r_1/\{q_1, q_2\}$} node[below] {$r_2/\{q_1, q_2\}$} (metaf);
	\draw[t,loop above] (metaf) to node[right,very near end] {$a$} (metaf);
\end{tikzpicture}
	}\fi

		\caption{A language-changing quotienting.}
		\label{fig:auto_merge_stack_symbols_2}
	\end{subfigure}
	\caption{A \vpa where quotienting of states leads to quotienting of stack symbols.}
	\label{fig:auto_merge_stack_symbols}
\end{figure}
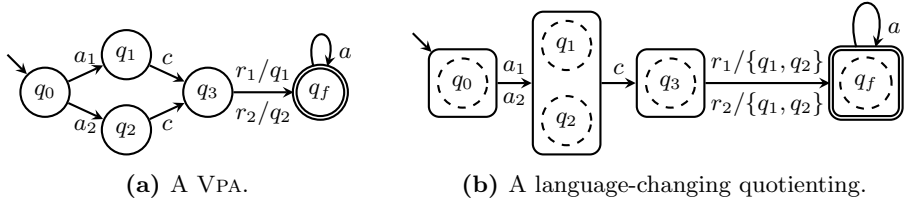

Consider the \vpa in Figure~\ref{fig:auto_merge_stack_symbols_1}.
Since for (weakly-hierarchical) \vpa, stack symbols are states, merging the states $q_1$ and $q_2$ implicitly merges the stack symbols $q_1$ and $q_2$ as well.
After merging we receive the \vpa in Figure~\ref{fig:auto_merge_stack_symbols_2} which recognizes a different language (e.g., it accepts the word $a_1 c\, r_2$).

\subsection{Our approach}
We define an equivalence relation over \vpa states for quotienting that is language-preserving. 
This equivalence relation exploits our key observation, namely that we can merge states if they have the same behavior on equivalent reachable stacks, even if they have different behavior in general (Section~\ref{sec:quotienting}).
We show an encoding of such a relation as a Boolean partial maximum satisfiability (\pmaxsat) instance (Section~\ref{sec:maxsat}).
In order to solve these instances efficiently, we propose a greedy algorithm that finds suboptimal solutions (Section~\ref{sec:implementation}).
As a proof of concept, we implemented the algorithm and evaluated it in the context of 
the automata-based software verifier \ultimate{} \automizer~\cite{HeizmannDGLMSP16,HeizmannHP13} (Section~\ref{sec:experiments}).

\section{Visibly pushdown automata}\label{sec:automata}

In this section, we recall the basic definitions for visibly pushdown automata~\cite{AlurMadhusudan:2004} and quotienting.
After that, we characterize when an automaton is \emph{live}.

\subsection{Preliminaries}

\subsubsection{Alphabet}
A \emph{(visibly pushdown) alphabet} $\AlphaAll = \AlphaI \uplus \AlphaC \uplus \AlphaR$ is a partition consisting of three finite sets of \emph{internal} (\AlphaI), \emph{call} (\AlphaC), and \emph{return} (\AlphaR) symbols.
A \emph{word} is a sequence of symbols.
We denote the set of finite words over alphabet \AlphaAll by $\AlphaAll^*$ and the empty word by \eps.
As a convention we use $a$ for internal, $c$ for call, and $r$ for return symbols, \symb for any type of symbol, and $v, w$ for words.

\smallskip

The set of \emph{well-matched} words over \AlphaAll, \wellmatched, is the smallest set satisfying:
\begin{enumin}{;}{and}
	\item $\eps \in \wellmatched$
	
	\item if $w \in \wellmatched$, so is $wa$ for $a \in \AlphaI$
	
	\item if $v, w \in \wellmatched$, so is $vcwr$ for $c\,r \in \AlphaC \cdot \AlphaR$, and we call symbols $c$ and $r$ \emph{matching}
\end{enumin}
Given a word over \AlphaAll, for any return symbol we can uniquely determine whether the symbol is matching.
The set of \emph{matched-return} words, \matchedreturn, consists of all words where each return symbol is matching.
Clearly, $\wellmatched$ is a subset of $\matchedreturn$.

\subsubsection{Visibly pushdown automaton}
A \emph{visibly pushdown automaton} (\vpa) is a tuple $\A = (Q, \AlphaAll, \bot, \transAll, Q_0, F)$ with a finite set of states $Q$, a visibly pushdown alphabet \AlphaAll, a bottom-of-stack symbol $\bot \notin Q$, a transition relation $\transAll = (\transI, \transC, \transR)$ consisting of internal transitions $\transI \subseteq Q \times \AlphaI \times Q$, call transitions $\transC \subseteq Q \times \AlphaC \times Q$, and return transitions $\transR \subseteq Q \times \AlphaR \times Q \times Q$, a nonempty set of initial states $Q_0 \subseteq Q$, and a set of accepting states $F \subseteq Q$.

\smallskip

A \emph{stack} \stack is a word over $\stacks \defeq \{\bot\} \cdot Q^*$.
We write $\stack[i]$ for the $i$-th symbol of \stack.
%
A \emph{configuration} is a pair $(q, \stack) \in Q \times \stacks$.
%
A
\emph{run} $\run_\A(w)$ of \vpa \A on word $w = \symb_1 \symb_2 \cdots \in \AlphaAll^*$ is a sequence of configurations $(q_0, \stack_0) (q_1, \stack_1) \cdots$ according to the following rules (for $i \geq 0$):
\begin{enumerate}
	\item If $\symb_{i+1} \in \AlphaI$ then $(q_i, \symb_{i+1}, q_{i+1}) \in \transI$ and $\stack_{i+1} = \stack_i$.
	
	\item If $\symb_{i+1} \in \AlphaC$ then $(q_i, \symb_{i+1}, q_{i+1}) \in \transC$ and $\stack_{i+1} = \stack_i \concat q_i$.
	
	\item If $\symb_{i+1} \in \AlphaR$ then $(q_i, \symb_{i+1}, \hier{q}, q_{i+1}) \in \transR$ and $\stack_i = \stack_{i+1} \concat \hier{q}$.
\end{enumerate}

A run is \emph{initial} if $(q_0, \stack_0) \in Q_0 \times \{\bot\}$.
%
A configuration $(q, \stack)$ is \emph{reachable} if there exists some initial run $\run = (q_0, \stack_0)\linebreak[1](q_1, \stack_1) \cdots$ such that $(q_i, \stack_i) = (q, \stack)$ for some $i \geq 0$, and \emph{unreachable} otherwise.
Similarly, we say that a stack \stack is reachable (\resp unreachable) for state $q$ if $(q, \stack)$ is reachable (\resp unreachable).
%
A run of length $n$ is \emph{accepting} if $q_n \in F$.
A word $w \in \AlphaAll^*$ is \emph{accepted} if some initial run $\run_\A(w)$ is accepting.
The \emph{language} recognized by a \vpa \A is defined as $\lang(\A) \defeq \{w \mid w \text{ is accepted by } \A\}$.
%
A \vpa is \emph{deterministic} if it has one initial state and the transition relation is functional.

A \emph{finite automaton} (\fa) is a \vpa where $\AlphaC = \AlphaR = \emptyset$.

\begin{remark}
	We use a variant of \vpa that deviates from the \vpa model by Alur and Madhusudan~\cite{AlurMadhusudan:2004} in two ways:
	\begin{enumin}{.}{}
		\item We forbid return transitions when the stack is empty, i.e., the automata accept only matched-return words; this assumption is also used in other works~\cite{Srba:2009,KumarEtal:2006}%
		
		\item We consider \emph{weakly-hierarchical} \vpa where a call transition implicitly pushes the current state on the stack; this assumption is also a common assumption~\cite{KumarEtal:2006,ChervetWalukiewicz:2007}; every \vpa can be converted to weakly-hierarchical form with $2 |Q| |\AlphaAll|$ states~\cite{AlurMadhusudan:2009}
	\end{enumin}
	
	Both assumptions are natural in the context of computer programs:
	The call stack can never be empty, and return transitions always lead back to the respective program location after the corresponding call.
\end{remark}

\subsubsection{Quotienting}

For an equivalence relation over some set $S$, we denote the equivalence class of element $e$ by $[e]$; analogously, lifted to sets, let $[T] \defeq \{[e] \mid e \in T\}$.

Given a \vpa
	$\A = (Q, \AlphaAll, \bot, (\transI, \transC, \transR), Q_0, F)$
and an equivalence relation $\eqrel \, \subseteq Q \times Q$ on states, the \emph{quotient \vpa} is the \vpa
	$\A/_\eqrel \defeq ([Q], \AlphaAll, \bot, \transAll', [Q_0], [F])$
with $\transAll' = (\transI', \transC', \transR')$ and
\begin{itemize}
	\myitem $\transI' = \{([p], a, [p']) \mid \exists (q, a, q') \in \transI \quantDot q \in [p], q' \in [p']\}$,
	
	\myitem $\transC' = \{([p], c, [p']) \mid \exists (q, c, q') \in \transC \quantDot q \in [p], q' \in [p']\}$, and
	
	\myitem $\transR' = \{([p], r, [\hier{p}], [p']) \mid \exists (q, r, \hier{q}, q') \in \transR \quantDot q \in [p], q' \in [p']$, $\hier{q} \in [\hier{p}] \}$.
\end{itemize}

\smallskip
\goodbreak

\emph{Quotienting} is the process of merging states from the same equivalence class to obtain the quotient \vpa; this implicitly means merging stack symbols, too.

\subsection{Live visibly pushdown automata}

Let $\AlphaS \defeq Q \cup \{\bot\}$ be the \emph{stack alphabet}.
The function $\Top : \stacks \to \AlphaS$ returns the topmost symbol of a stack:
\begin{equation*}
	\Top(\stack) \defeq
	\begin{cases}
		\bot & \stack = \bot \\
		q & \stack = \stack' \concat q \text{ for some } \stack' \in \stacks
	\end{cases}
\end{equation*}

Given a state $q$, the function $\tops : Q \to 2^{\AlphaS}$ returns the topmost symbols of all reachable stacks \stack for $q$ (i.e., reachable configurations $(q, \stack)$):
\begin{equation*}
	\tops(q) \defeq \{ \Top(\stack) \mid \exists \stack \in \stacks \quantDot (q, \stack) \text{ is reachable} \}
\end{equation*}

For seeing that \tops is computable, consider a \vpa $\A = (Q, \AlphaAll, \bot, \transAll, Q_0, F)$.
The function \tops is the smallest function $f : Q \to 2^{\AlphaS}$ satisfying:
\begin{enumerate}
	\item $q \in Q_0 \implies \bot \in f(q)$
	
	\item $\hier{q} \in f(q), (q, a, q') \in \transI \implies \hier{q} \in f(q')$
	
	\item $(q, \stack)$ reachable for some \stack, $(q, c, q') \in \transC \implies q \in f(q')$
	
	\item $\hier{q} \in f(q), (q, r, \hier{q}, q') \in \transR \implies f(\hier{q}) \subseteq f(q')$
\end{enumerate}

We call a \vpa \emph{live} if the following holds.
For each state $q$ and for each internal and call symbol $\symb$ there is at least one outgoing transition $(q, \symb, q')$ to some state $q'$;
additionally, for each return symbol $r$ and state \hier{q} there is at least one outgoing return transition $(q, \hier{q}, r, q')$ to some state $q'$ if and only if $\hier{q} \in \tops(q)$.

Note that a live \vpa has a total transition relation in a weaker sense:
There are outgoing return transitions from state $q$ if and only if the respective transition can be taken in at least one run.
That is, we forbid return transitions when no corresponding configuration is reachable.
Every \vpa can be converted to live form by adding one sink state.

\begin{remark}
	For live \vpa \A, a run $\run_\A(w)$ on word $w$ can only ``get stuck'' in an empty-stack configuration, i.e., if $w = v_1 r \, v_2$ with $r \in \AlphaR$ such that $\run_\A(v_1) = (q_0, \stack_0) \cdots (q_k, \bot)$ for some $q_k \in Q$.
	If $w \in \matchedreturn$, no run gets stuck.
\end{remark}

For the remainder of the paper, we fix a live \vpa $\A = (Q, \AlphaAll, \bot, \transAll, Q_0, F)$.
We sometimes refer to this \vpa as the \emph{input automaton}.

\section{A quotienting relation for VPA}\label{sec:quotienting}

In this section, we define an equivalence relation on the states of a \vpa that is useful for quotienting, i.e., whose respective quotient \vpa is language-preserving.

\smallskip

We first need the notion of \emph{closure under successors} for each kind of symbol.

\smallskip

Let $R \subseteq Q \times Q$ be a binary relation over states and let $p, q, \hier{p}, \hier{q} \in Q$ be states. We say that $R$ is
\begin{itemize}
	\myitem \emph{closed under internal successors for $(p, q)$} if for each internal symbol $a \in \AlphaI$
	\begin{itemize}
		\item for all $(p, a, p') \in \transI$ there exists $(q, a, q') \in \transI$ s.t. $(p', q')\in R$ and
		
		\item for all $(q, a, q') \in \transI$ there exists $(p, a, p') \in \transI$ s.t. $(p', q')\in R$,
	\end{itemize}
	
	\myitem \emph{closed under call successors for $(p, q)$} if for each call symbol $c \in \AlphaC$
	\begin{itemize}
		\item for all $(p, c, p') \in \transC$ there exists $(q, c, q') \in \transC$ s.t. $(p', q')\in R$ and
		
		\item for all $(q, c, q') \in \transC$ there exists $(p, c, p') \in \transC$ s.t. $(p', q')\in R$,
	\end{itemize}
	
	\myitem \emph{closed under return successors for $(p, q, \hier{p}, \hier{q})$} if for each return symbol $r \in \AlphaR$
	\begin{itemize}
		\item for all $(p, r, \hier{p}, p') \in \transR$ there exists $(q, r, \hier{q}, q') \in \transR$ s.t. $(p', q')\in R$ and
		
		\item for all $(q, r, \hier{q}, q') \in \transR$ there exists $(p, r, \hier{p}, p') \in \transR$ s.t. $(p', q')\in R$.
	\end{itemize}
\end{itemize}

We are ready to present an equivalence relation that is useful for quotienting using a fixpoint characterization.

\begin{definition}[Reachability-aware quotienting relation]\label{def:merge_rel}
	Let \A be a \vpa and $R\subseteq Q\times Q$ be an equivalence relation over states.
	We say that $R$ is a \emph{\Qrel} if for each pair of states $(p, q)\in R$ the following constraints hold.
	\begin{enumerate}[label=(\roman*),ref=(\roman*)]
		\item \label{item:Erel_constraint_accept} State $p$ is accepting if and only if state $q$ is accepting $(p \in F \iff q \in F)$.
		
		\item \label{item:Erel_constraint_internal} $R$ is closed under internal successors for $(p, q)$.
		
		\item \label{item:Erel_constraint_call} $R$ is closed under call successors for $(p, q)$.
		
		\item \label{item:Erel_constraint_return} For each pair of states (\resp topmost stack symbols) $(\hier{p}, \hier{q}) \in R$,
		\begin{itemize}
			\myitem $R$ is closed under return successors for $(p, q, \hier{p}, \hier{q}$), or
			
			\myitem no configuration $(q, \stack_q)$ with $\hier{q} = \Top(\stack_q)$ is reachable, or
			
			\myitem no configuration $(p, \stack_p)$ with $\hier{p} = \Top(\stack_p)$ is reachable.
		\end{itemize}
	\end{enumerate}
\end{definition}

\begin{remark}
	``No configuration $(q, \stack_q)$ with $\hier{q} = \Top(\stack_q)$ is reachable'' is equivalent to ``$\hier{q} \notin \tops(q)$''.
	The equality relation $\{ (q, q) \mid q \in Q \}$ is a \Qrel for any \vpa; the respective quotient \vpa is isomorphic to the input automaton.
\end{remark}

\begin{example}
	Consider again the \vpa from Figure~\ref{fig:auto_two_out_of_three_1}.
	We claim that the relation
	$
		R \defeq \{ (q, q) \mid q \in Q \} \cup \{ (q_1, q_2), (q_2, q_1) \}
	$
	is a \Qrel.
	Note that it corresponds to the quotient \vpa from Figure~\ref{fig:auto_two_out_of_three_2}.
	First we observe that $R$ is an equivalence relation.
	We check the remaining constraints only for the two pairs $(q_1, q_2)$ and $(q_2, q_1)$.
	Both states are not accepting.
	Relation $R$ is closed under internal (here: $a$) and call (here: none, i.e., implicitly leading to a sink) successors.
	The return transition constraint is satisfied because in state $q_2$ no stack with topmost symbol $q_0$ is reachable ($q_0 \notin \tops(q_2)$).
\end{example}

\smallskip

We want to use a \Qrel for language-preserving quotienting.
For this purpose we need to make sure that unreachable configurations in Definition~\ref{def:merge_rel} do not enable accepting runs that are not possible in the original \vpa.
In the remainder of this section, we show that this is indeed the case.

\smallskip

Given an equivalence relation $R \subseteq Q\times Q$ on states, we call a stack $\stack$ the \emph{$R$-quotienting} of some stack $\stack'$ of the same height if either $\stack = \stack' = \bot$ or for all $i = 2, \dots, |\stack|$ each symbol $\stack[i]$ is the equivalence class of $\stack'[i]$, i.e., $\stack'[i] \in [\stack[i]]$.
We write $\stack' \in [\stack]$ in this case.
(We compare stacks only for $i \geq 2$ because the first stack symbol is always $\bot$.)

\begin{lemma}[Corresponding run]\label{lem:correspondingRun}
	Let \A be a \vpa and \eqrel be some \Qrel for \A.
	Then for any matched-return word $w$ and respective run
	\begin{equation*}
		\run_{\A/\eqrel}(w) = ([p_0], \bot) \cdots ([p_n], [\stack_n])
	\end{equation*}
	with $p_0 \in Q_0$ in $\A/_\eqrel$ there is some corresponding run
	\begin{equation*}
		\run_{\A}(w) = (q_0', \bot) \cdots (q_n', \stack_n')
	\end{equation*}
	in \A such that $q_i' \in [p_i]$ and $\stack_i' \in [\stack_i]$ for all $i \geq 0$, and furthermore $q_0' \in Q_0$.
\end{lemma}
\begin{proof}
	The proof is by induction on the length of $w$.
	The case for $w = \eps$ is trivial.
	Now assume $w' = w \concat \symb$ for $\symb \in \AlphaAll$ and fix some run $\run_{\A/\eqrel}(w') = ([p_0], \bot) \cdots ([p_n], [\stack_n]) \concat ([p_{n+1}], [\stack_{n+1}])$.
	The hypothesis ensures that there exists a corresponding run for the prefix $\run_{\A}(w) = (q_0', \bot) \cdots (q_n', \stack_n')$ s.t.\ $q_n' \in [p_n]$ and $\stack_n' \in [\stack_n]$.
	We will extend this run in each of the three cases for symbol $\symb$.
	
	\smallskip
	
	1) If $\symb \in \AlphaI$, then, since there is a transition $([p_n], \symb, [p_{n+1}]) \in \transI{}_{/\eqrel}$, there exist some states $q_n'' \in [p_n]$ and $q_{n+1}'' \in [p_{n+1}]$ s.t.\ $(q_n'', \symb, q_{n+1}'') \in \transI$ (from the definition of quotienting).
	Using that \eqrel is closed under internal successors, there also exists a target state $q_{n+1}' \in [p_{n+1}]$ s.t.\ $(q_n', \symb, q_{n+1}') \in \transI$.
	Additionally, because $\symb \in \AlphaI$, we have that $\stack_{n+1}' = \stack_n' \in [\stack_n] = [\stack_{n+1}]$ by the hypothesis.
	
	2) If $\symb \in \AlphaC$, a similar argument holds, only this time the stack changes.
	We have that $\stack_{n+1}' = \stack_n' \concat q_n' \in [\stack_n \concat p_n] = [\stack_{n+1}]$ by the hypothesis.
	
	3) If $\symb \in \AlphaR$, then the configuration $(q_n', \stack_n')$ is reachable (witnessed by the run $\run_{\A}(w)$).
	Since \eqrel is closed under return successors for all states in $[p_n]$ (modulo unreachable configurations), for each top-of-stack symbol $\hier{q} \in [\Top(\stack_n')]$ s.t.\ $(q_n', \stack'' \concat \hier{q})$ is reachable for some stack $\stack''$ there exists a corresponding return transition $(q_n', \symb, \hier{q}, q_{n+1}') \in \transR$ with $q_{n+1}' \in [p_{n+1}]$; in particular, this holds for $\hier{q} = \Top(\stack_n')$.
	Recall that \A is assumed to be live, which ensures that every return transition that exists in the quotient \vpa has such a witness.
	The stack property $\stack_{n+1}'\in [\stack_{n+1}]$ follows from the hypothesis.
	\qed
\end{proof}

From the above lemma we can conclude that quotienting with a \Qrel preserves the language.

\goodbreak

\begin{theorem}[Language preservation of quotienting]
	Let \A be a \vpa and let \eqrel be a \Qrel on the states of \A.
	Then $L(\A) = L(\A/_\eqrel)$.
\end{theorem}
\begin{proof}
	Clearly, $\lang(\A) \subseteq \lang(\A/_\eqrel)$ for any equivalence relation \eqrel.
	We show the other inclusion by means of a contradiction.
	
	Assume there exists a word $w$ s.t.\ $w \in L(\A/_\eqrel) \setminus L(\A)$.
	By assumption, in $\A/_\eqrel$ there is an initial accepting run $\run_{\A/_\eqrel}(w)$.
	Then, by Lemma~\ref{lem:correspondingRun}, there is a corresponding run $\run_\A(w)$, and furthermore this run is initial.
	
	The run $\run_\A(w)$ is also accepting by the property that $[p] \in [F]$ if and only if $q \in F$ for all $q \in [p]$ (cf. Property~\ref{item:Erel_constraint_accept} of a \Qrel).
	\qed
\end{proof}

\section{Computing quotienting relations}\label{sec:maxsat}

In Section~\ref{sec:quotienting}, we introduced the notion of a \Qrel and showed how we can use it to minimize \vpa while preserving the language.
In this section, we show how we can compute a \Qrel.
For this purpose, we provide an encoding as a partial maximum satisfiability problem (\pmaxsat).
From a (in fact, any) solution, i.e., satisfying assignment, we can synthesize a \Qrel.
While this does not result in the coarsest \Qrel possible, the relation obtained is locally optimal, i.e., there is no coarser \Qrel that is a strict superset.

\subsection{Computing RAQ relations}

Note that in general there are many possible instantiations of a \Qrel, e.g., the trivial equality relation which is not helpful for minimization.
Since we are interested in reducing the number of states, we prefer coarser relations over finer relations.

To obtain a coarse relation, we describe an encoding of the \Qrel constraints as an instance of the \pmaxsat problem~\cite{ChaIKM97,FuMalik:2006}.
Such a problem consists of a propositional logic formula in conjunctive normal form with each clause being marked as either \emph{hard} or \emph{soft}.
The task is to find a truth assignment such that all hard clauses are satisfied and the number of the satisfied soft clauses is maximal.

\mysubsubsection{SAT encoding}

For the moment, we ignore soft clauses and provide a standard \sat encoding of the constraints.
The encoding has the property
that any satisfying assignment induces a valid \Qrel \eqrel.

Let \true and \false be the Boolean constants.
We need $\On{n^2}$ variables of the form $\X{p, q}$ where $p$ and $q$ are states of the input automaton.
The idea is that $p \eqrel q$ holds if we assign the value \true to \X{p, q}.
(We ignore the order of $p$ and $q$ as \eqrel must be symmetric.)
We express the constraints from Definition~\ref{def:merge_rel} as follows.

\smallskip

Consider the constraint~\ref{item:Erel_constraint_accept}.
For each pair of states $(p, q)$ not satisfying the constraint we introduce the clause
\begin{equation}\label{eq:clauseNonsimulation}
	\neg \X{p, q}.
\end{equation}

\smallskip

Consider the constraints~\ref{item:Erel_constraint_internal}, \ref{item:Erel_constraint_call}, \ref{item:Erel_constraint_return}.
For each transition $(p, a, p') \in \transI$, $(p, c, p') \in \transC$, and $(p, r, \hier{p}, p') \in \transR$ and all states $q$ and $\hier{q}$ we respectively construct one of the following clauses.
\begin{align}
	\neg \X{p, q} &\lor (\X{p', q_1^a} \lor \dots \lor \X{p', q_{k_a}^a})
		\label{eq:clauseITransitions} \\
	\neg \X{p, q} &\lor (\X{p', q_1^c} \lor \dots \lor \X{p', q_{k_c}^c})
		\label{eq:clauseCTransitions} \\
	\neg \X{p, q} &\lor \neg \X{\hier{p}, \hier{q}} \lor (\X{p', q_1^r} \lor \dots \lor \X{p', q_{k_r}^r})
		\label{eq:clauseRTransitions} 
\end{align}
Here the $q_i^a$/$q_i^c$ are the respective $a$/$c$-successors of $q$ and the $q_i^r$ are the $r$-successors of $q$ with stack symbol $\hier{q}$.
To account for the unreachable configuration relaxation, we may omit return transition clauses~\eqref{eq:clauseRTransitions} where $\hier{p} \notin \tops(p)$ or $\hier{q} \notin \tops(q)$.

We also need to express that \eqrel is an equivalence relation, i.e., we need additional reflexivity clauses
\begin{equation}\label{eq:clauseReflexivity}
	\X{q_1, q_1}
\end{equation}
and transitivity clauses
\begin{equation}\label{eq:clauseTransitivity}
	\neg \X{q_1, q_2} \lor \neg \X{q_2, q_3} \lor \X{q_1, q_3}
\end{equation}
for any distinct states $q_1$, $q_2$, $q_3$ (assuming there are least three states).
Recall that our variables already ensure symmetry.

\smallskip

Let \cnf be the conjunction of all clauses of the form~\eqref{eq:clauseNonsimulation}, \eqref{eq:clauseITransitions}, \eqref{eq:clauseCTransitions}, \eqref{eq:clauseRTransitions}, \eqref{eq:clauseReflexivity}, and \eqref{eq:clauseTransitivity}.
All assignments satisfying \cnf represent valid \Qrel{}s.

However, we know that the assignment
\begin{equation*}
	\X{p, q} \mapsto
	\begin{cases}
		\true & p = q \\
		\false & \text{otherwise}
	\end{cases}
\end{equation*}
corresponding to the equality relation is always trivially satisfying.
Such an assignment is not suited for our needs.
We consider an assignment \emph{optimal} if it represents a \Qrel with a coarsest partition.

\mysubsubsection{PMax-SAT encoding}
We now describe an extension of the \sat encoding to a \pmaxsat encoding.
In this setting, we can enforce that the number of variables that are assigned the value \true is maximal.

\smallskip

As an addition to \cnf, we add for every two states $p$, $q$ with $p \neq q$ the clause
\begin{align}
	\X{p, q}
		\label{eq:clauseSingleton}
\end{align}

and finally we consider all old clauses, i.e., clauses of the form~\eqref{eq:clauseNonsimulation}--\eqref{eq:clauseTransitivity}, as hard clauses and all clauses of the form~\eqref{eq:clauseSingleton} as soft clauses.

\subsection{Locally maximal RAQ relation}\label{sec:local_optimum}
Note that an assignment obtained from the \pmaxsat encoding does not necessarily give us a coarsest \Qrel.
Consider a \vpa with seven states $q_0, \dots, q_6$ and the partition $\big\{ \{q_0, q_1, q_2, q_3\}$, $\{q_4\}$, $\{q_5\}$, $\{q_6\} \big\}$.
Here we set six variables to \true (all pairs of states from the first set).
However, the partition $\big\{ \{q_0, q_1, q_2\}$, $\{q_3, q_4\}$, $\{q_5, q_6\} \big\}$ is coarser, and yet we only set five variables to \true.

\medskip

Despite not finding the globally maximal solution, we can establish local maximality.
\begin{theorem}[Local maximum]\label{thm:local_optimum}
	A satisfying assignment of the \pmaxsat instance corresponds to a \Qrel such that no strict superset of the relation is also a \Qrel.
\end{theorem}
\begin{proof}
	It is clear from the construction that in the obtained assignment, no further variable \X{p, q} can be assigned the value \true.
	Each such variable determines membership of the symmetric pairs $(p, q)$ and $(q, p)$ in the \Qrel.
	\qed
\end{proof}

\section{Experimental evaluation}\label{sec:evaluation}

In this section, we report on our implementation and its potential in practice.

\subsection{Implementation}\label{sec:implementation}

Initially, we apply the following preprocessing steps for reducing the complexity.
First, we remove unreachable and dead states and make the \vpa live for return transitions (we do not require that the \vpa is total for internal or call transitions).
Second, we immediately replace variables $\X{p}$ by \true (reflexivity).
Third, we construct an initial partition of the states and replace variables $\X{p,q}$ by \false if $p$ and $q$ are not in the same block.
This partition is the coarsest fixpoint of a simple partition refinement such that states in the same block have the same acceptance status, the same outgoing internal and call symbols, and, if all states in a block have a unique successor under an internal/call symbol, those successors are in the same block (cf.\ Definition~\ref{def:merge_rel} and Hopcroft's algorithm~\cite{Hopcroft:1971}).

\smallskip

Optimally solving a \pmaxsat instance is an \np-complete problem.
Expectedly, a straightforward implementation of the algorithm presented in Section~\ref{sec:maxsat} using an off-the-shelf \pmaxsat solver does not scale to interesting problems (see also \iftacasversion{the extended version~\cite{extendedVersionArxiv}\else{Appendix~\ref{app:external_solver}}\fi\xspace).
Therefore, we implemented a domain-specific greedy \pmaxsat solver that only maximizes the satisfied soft clauses locally.

Our solver is interactive, i.e., clauses are added one after another, and propagation is applied immediately.
After adding the last clause, the solver chooses some unset variable and first sets it to \true optimistically.
Theorem~\ref{thm:local_optimum} still holds with this strategy.
Apart from that, the solver follows the standard DPLL algorithm and uses no further enhancements found in modern \sat solvers.

\begin{remark}
	If the \vpa is deterministic, we obtain a Horn clause system.
	Then the above algorithm never needs to backtrack for more than one level, as the remaining clauses can always be satisfied by assigning \false to the variables.
\end{remark}

\goodbreak

The main limitation of the approach is the memory consumption.
Clearly, the majority of clauses are those expressing transitivity.
Therefore, we implemented and integrated a solver for the theory of equality:
When a variable $\X{p, q}$ is set to \true, this solver returns all variables that must also be set to \true for consistency.
That allowed us to omit the transitivity clauses (see\iftacasversion{~\cite{extendedVersionArxiv}}\else{ Appendix~\ref{app:transitivity_checker}}\fi\xspace for details).

\subsection{Experiments}\label{sec:experiments}

Our evaluation consists of three parts.
First, we evaluate the impact of our minimization on an application, namely the software verifier \ultimate{} \automizer.
Second, we evaluate the performance of our minimization on automata that were produced by \ultimate{} \automizer.
Third, we evaluate the performance of our minimization on a set of random automata.
All experiments are performed on a PC with an Intel i7 3.60 GHz CPU running Linux.

\mysubsubsection{Impact on the software verifier \ultimate{} \automizer}
The software verifier \ultimate{} \automizer~\cite{HeizmannDGLMSP16}
follows an automata-based approach~\cite{HeizmannHP13}
in which sets of program traces are represented by automata.
The approach can be seen as a \cegar-style algorithm in which an abstraction is iteratively refined.
This abstraction is represented as a weakly-hierarchical \vpa where the automaton stack only keeps track of the states from where function calls were triggered.

For our evaluation, we run \ultimate{} \automizer on a set of C programs in two different modes.
In the mode ``No minimization'' no automata minimization is applied.
In the mode ``Minimization'' we apply our minimization in each iteration of the \cegar loop to the abstraction if it has less than $10{,}000$ states. (In cases where the abstraction has more than $10{,}000$ states the minimization can be too slow to pay off on average.)

As benchmarks we took C programs from the repository of the \svcomp 2016~\cite{Beyer16} and let \ultimate{} \automizer analyze if the error location is reachable.
In this repository the folders \texttt{systemc} and \texttt{eca-rers2012} contain programs that use function calls (hence the \vpa contain calls and returns) and in whose analysis the automata sizes are a bottleneck for \ultimate{} \automizer. 
We randomly picked 100 files from the \texttt{eca-rers2012} folder and took all 65 files from the \texttt{systemc} folder.
The timeout of \ultimate{} \automizer was set to $300$~s and the available memory was restricted to $4$~GiB.

\begin{table}[t]
	\caption{%
		Performance of \ultimate{} \automizer with and without minimization.
		Column \# shows the number of successful reachability analyses (out of $165$), average run time is given in milliseconds, average removal shows the states removed for all iterations, and the last column shows the relative number of iterations where minimization was employed.
		The first two rows contain the data for those programs where both modes succeeded, and the third row contains the data for those programs where only the minimization mode succeeded.
	}
	\label{tab:benchmarks_automizer_online}
	\vspace*{2mm}
	\centering
	%
	\begin{tabular}{|c|c|c|r|c|c|c|}%
	\hline
	Mode & Set & \colhead{\,\#\,} & \begin{minipage}{10mm}{\vspace*{.5mm}\hspace*{-1.5mm}$\varnothing$ time \\ \hspace*{0.5mm}total\vspace*{1mm}}\end{minipage} & \begin{minipage}{19mm}{\vspace*{.5mm}\centering$\varnothing$ time \\ minimization\vspace*{1mm}}\end{minipage} & $\varnothing$ removal & \begin{minipage}{26mm}{\vspace*{.5mm}\centering \% iterations \\with minimization\vspace*{1mm}}\end{minipage}%
	\\ \hline
	No minimization & \multirow{2}{*}{both} & \multirow{2}{*}{66} & \hspace*{.1mm} 16085 \hspace*{.1mm} & -- & -- & -- \\
	Minimization & & & \hspace*{.1mm} 15564 \hspace*{.1mm} & \, 2649 & 3077 & 75 \\
	\hline
	Minimization & exclusive & 12 & \hspace*{.1mm} 101985 \hspace*{.1mm} & 61384 & 8472 & 76 \\
	\hline
	\end{tabular}
	%
\end{table}

The results are given in Table~\ref{tab:benchmarks_automizer_online}.
Our minimization increases the number of programs that are successfully analyzed from $66$ to $78$.
On programs that are successfully analyzed in both modes, the mode using minimization is slightly faster. 
Hence, the additional cost due to minimization is more than compensated by savings in other operations on the (now smaller) \vpa on average.

\mysubsubsection{Evaluation on automata from \ultimate{} \automizer{}}
To evaluate the performance of our minimization algorithm in more details, we applied it to a benchmark set that consists of $1026$ \vpa produced by \ultimate{} \automizer.
All automata from this set contain call and return transitions and do not contain any dead ends (states from which no accepting state is reachable).
Details on the construction of these automata can be found in \iftacasversion{the extended version~\cite{extendedVersionArxiv}\else{Appendix~\ref{app:automata}}\fi\xspace.

\begin{figure}[t!]
	\hspace*{3mm}
	\begin{minipage}{0.4\textwidth}
		\tikzsetnextfilename{plot_offline_states}
		\scatterPlot[legend entries={100\%,50\%,data D,data N},legend to name=scatter_legend,]{\# states (input)}{output}{AutomizerOfflineDet}{STATES_INPUT}{STATES_OUTPUT}{35000}{{0,10000,20000,30000}}{AutomizerOfflineNondet}
	\end{minipage}
	\hspace*{15.6mm}
	\begin{minipage}{0.4\textwidth}
		\tikzsetnextfilename{plot_offline_internal_transitions}
		\scatterPlot{\# internal transitions (input)}{output}{AutomizerOfflineDet}{TRANSITIONS_INTERNAL_INPUT}{TRANSITIONS_INTERNAL_OUTPUT}{36000}{{0,10000,20000,30000}}{AutomizerOfflineNondet}
	\end{minipage}
	\hspace*{3mm}
	
	\vspace*{-12mm}
	\begin{center}
		\hspace*{-5mm} \iftacasversion{\tikzexternaldisable \ref{scatter_legend} \tikzexternalenable}\else{%
	\ifextendedversion{%
		\includegraphics{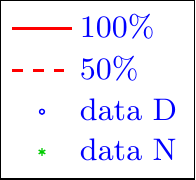}
	}\else{%
		\input{tikz/main-figure_crossref0}
	}\fi
}\fi
	\end{center}
	\vspace*{-10mm}
	
	\hspace*{3mm}
	\begin{minipage}{0.4\textwidth}
		\iftacasversion{%
		\tikzsetnextfilename{plot_offline_call_transitions}
		\scatterPlot[scaled ticks=base 10:-3]{\# call transitions (input)}{output}{AutomizerOfflineDet}{TRANSITIONS_CALL_INPUT}{TRANSITIONS_CALL_OUTPUT}{5000}{{0,1000,...,5000}}{AutomizerOfflineNondet}
		}\else{%
	\ifextendedversion{%
		\includegraphics{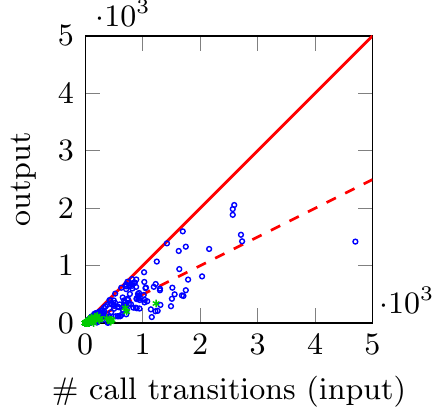}
	}\else{%
		\input{tikz/plot_offline_call_transitions}
	}\fi

		}\fi
	\end{minipage}
	\hspace*{15.6mm}
	\begin{minipage}{0.4\textwidth}
		\iftacasversion{%
		\tikzsetnextfilename{plot_offline_return_transitions}
		\scatterPlot[scaled ticks=base 10:-3]{\# return transitions (input)}{output}{AutomizerOfflineDet}{TRANSITIONS_RETURN_INPUT}{TRANSITIONS_RETURN_OUTPUT}{6200}{{0,1000,...,6000}}{AutomizerOfflineNondet}
		}\else{%
	\ifextendedversion{%
		\includegraphics{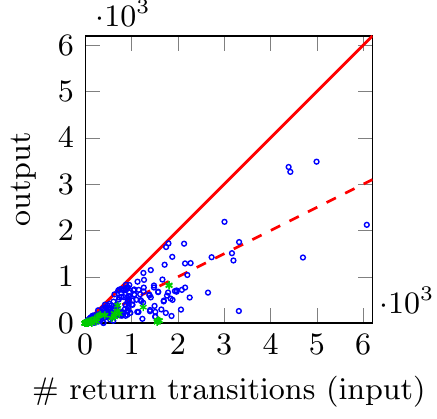}
	}\else{%
		\input{tikz/plot_offline_return_transitions}
	}\fi

		}\fi
	\end{minipage}
	\hspace*{3mm}
	
	\vspace*{2mm}
	
	\hspace*{1mm}
	\begin{minipage}{0.4\textwidth}
		\tikzsetnextfilename{plot_offline_states_vs_rel_reduction}
		\iftacasversion{%
		\dotPlotWithMeanHardCodedStates%
		[domain=1:35000,xmin=1,xmax=35000,xtick={0,10000,20000,30000},ymin=0,ymax=100,ytick={0,20,...,100},legend entries={mean D,mean N,data D,data N},legend to name=dot_legend,]{\# states (input)}{relative reduction}{AutomizerOfflineDet}{STATES_INPUT}{STATES_REDUCTION_RELATIVE}{\addlegendimage{line legend,blue,thick},\addlegendimage{line legend,darkgreen,thick}}{AutomizerOfflineNondet}%
		}\else{%
		%
	\ifextendedversion{%
		\includegraphics{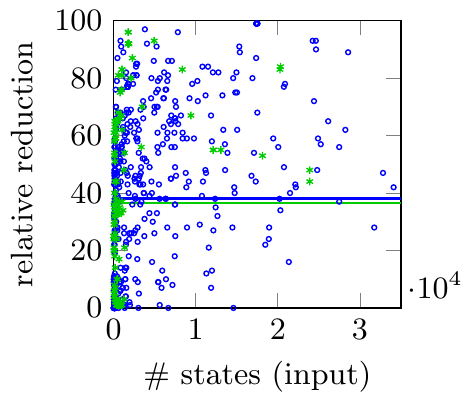}
	}\else{%
		\input{tikz/plot_offline_states_vs_rel_reduction}
	}\fi

		}\fi
	\end{minipage}
	\hspace*{-5.4mm} \iftacasversion{\tikzexternaldisable \ref{dot_legend} \tikzexternalenable}\else{%
	\ifextendedversion{%
		\includegraphics{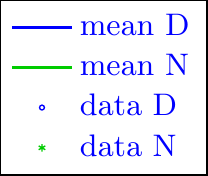}
	}\else{%
		\input{tikz/main-figure_crossref1}
	}\fi
}\fi 
	\begin{minipage}{0.4\textwidth}
		\hspace*{-1.135mm}
		\tikzsetnextfilename{plot_offline_transitions_vs_rel_reduction}
		\iftacasversion{%
		\dotPlotWithMeanHardCodedTransitions%
		[domain=1:37000,xmin=1,xmax=37000,xtick={0,10000,20000,30000},ymin=0,ymax=100,ytick={0,20,...,100}]{\# transitions (input)}{relative reduction}{AutomizerOfflineDet}{BUCHI_TRANSITIONS}{TRANSITIONS_REDUCTION_RELATIVE}{}{AutomizerOfflineNondet}%
		}\else{%
		%
	\ifextendedversion{%
		\includegraphics{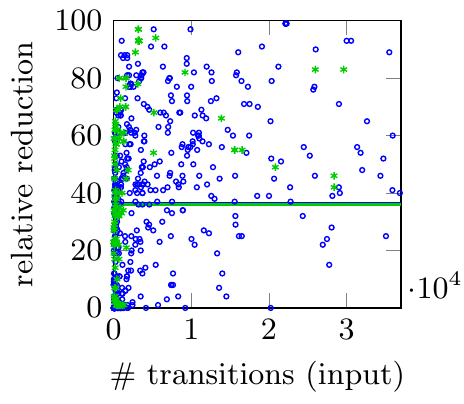}
	}\else{%
		\input{tikz/plot_offline_transitions_vs_rel_reduction}
	}\fi

		}\fi
	\end{minipage}
	\hspace*{2mm}
	
	\caption{%
		Minimization results on automata produced by \ultimate{} \automizer (see also Table~\ref{tab:benchmarks_automizer_offline}).
		D (N) stands for (non-)deterministic automata.
	}
	\label{fig:benchmarks_automizer_offline}
\end{figure}

\begin{table}[t]
	\caption{%
		Performance of our algorithm on automata produced by \ultimate{} \automizer (see also Figure~\ref{fig:benchmarks_automizer_offline}).
		We aggregate the data for all automata whose number of states is in a certain interval.
		Column \# shows the number of automata, \#nd shows the number of nondeterministic automata, and the other data is reported as average.
		The next seven columns show information about the input automata.
		The run time is given in milliseconds.
		The last two columns show the number of variables and clauses passed to the \pmaxsat solver.
	}
	\label{tab:benchmarks_automizer_offline}
	\vspace*{2mm}
	\centering
	\plottableOffline%
	{AutomizerOfflineAggregated}%
	{|c||r|r||r|r|r|r|r|r|r||r|r|r|}%
	{$|Q|$ (interval) & \colhead{\#} & \colheadDouble{\#nd} & \colhead{$|Q|$} & \colhead{$|\AlphaI|$} & \colhead{$|\AlphaC|$} & \colhead{$|\AlphaR|$} & \colhead{$|\transI|$} & \colhead{$|\transC|$} & \colheadDouble{$|\transR|$} & \colhead{time} & \colhead{$|\text{Var}|$} & \colhead{$|\text{Cls}|$}}%
	{\KEYaggregation & \KEYcount & \KEYnondeterminism & \KEYstatesIn & \KEYsymbolsInternal & \KEYsymbolsCall & \KEYsymbolsReturn & \KEYtransitionsInternalIn & \KEYtransitionsCallIn & \KEYtransitionsReturnIn & \KEYruntimeTotal & \KEYvariables & \KEYclauses}%
	{}
\end{table}

\smallskip

We ran our implementation on these automata using a timeout of $300$~s and a memory limit of $4$~GiB. Within the resource bounds we were able to minimize $596$ of the automata.
Details about these automata and the minimization run are presented in Table~\ref{tab:benchmarks_automizer_offline}.
In the table we grouped automata according to their size.
For instance, the first row aggregates the data of all automata that have up to 250 states.
The table shows that we were able to minimize automata up to a five-digit number of states and that automata that have a few thousand states can be minimized within seconds.
Figure~\ref{fig:benchmarks_automizer_offline} shows the sizes of the minimization results.
The first four graphs compare the sizes of input and output in terms of states and transitions.
The fourth graph shows that the (partly) significant size reduction is not only due to ``intraprocedural'' merges, but that also the number of return transitions is reduced.
The last two graphs show that the relative size reduction is higher on larger automata.
The reason is that small automata in \ultimate{} \automizer tend to have similarities to the control flow graph of a program, which is usually already minimal.

\mysubsubsection{Evaluation on random automata}

\begin{figure}[t!]
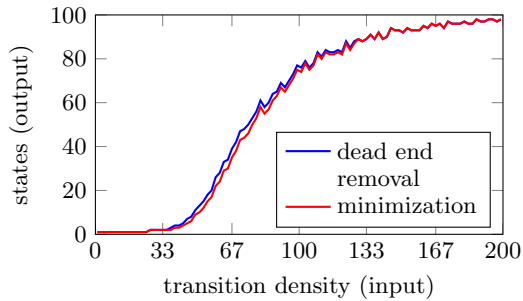

	\begin{center}
		\tikzsetnextfilename{plot_random_all}
		\doublelineplot[width=70mm,domain=0:200,xmin=0,xmax=200,xtick={0,33,67,100,133,167,200},ymin=0,ymax=100,ytick={0,20,40,60,80,100},legend entries={dead end,minimization},legend pos=south east]{transition density (input)}{states (output)}{SIZE_AFTER_PRE_PROC}{SIZE_OUTPUT}{\addlegendentry[text depth={}]{dead end}\addlegendimage{empty legend}\addlegendentry[text depth={}]{removal}}{\addlegendentry{minimization}}{Random}
	\end{center}
	%
	%
	\caption{%
		Minimization results on random \vpa with $100$~states, of which $50\%$ are accepting, and with one internal, call, and return symbol each. Return transitions are each inserted with $50$~random stack symbols.
		The transition density is increased in steps of $2\%$.
		Each data point stems from $500$~random automata.
	}
	\label{fig:benchmarks_random}
\end{figure}

The automata produced by \ultimate{} \automizer have relatively large alphabets 
(according to Table~\ref{tab:benchmarks_automizer_offline} there are on average less than 10 states per symbol)
and are extremely sparse (on average less than 1.5 transitions per state).
To investigate the applicability of our approach to \vpa without such structure, we also evaluate it on random nondeterministic \vpa.
We use a generalization of the random Büchi automata model by Tabakov and Vardi~\cite{TabakovV05} to \vpa (see \iftacasversion{the extended version~\cite{extendedVersionArxiv} for details\else{Appendix~\ref{app:random}}\fi\xspace).
Figure~\ref{fig:benchmarks_random} shows that our algorithm can remove some states on top of removing dead ends for lower transition densities, but overall it seems more appropriate to automata that have some structure.

\section{Related work}\label{sec:related_work}

Alur \etal~\cite{AlurEtal:2005} show that a canonical minimal \vpa does not exist in general.
They propose the \emph{single entry-\vpa} (\sevpa) model, a special \vpa of equivalent expressiveness with the following constraints:
Each state and call symbol is assigned to one of $k$ modules, and each module has a unique entry state which is the target of all respective call transitions.
This is enough structure to obtain the unique minimal $k$-\sevpa from any given $k$-\sevpa by quotienting.

Kumar \etal~\cite{KumarEtal:2006} extend the idea to \emph{modular \vpa}.
Here the requirement of having a unique entry per module is overcome, but more structure must be fixed to preserve a unique minimum -- most notably the restriction to weakly-hierarchical \vpa and the return alphabet being a singleton.

Chervet and Walukiewicz~\cite{ChervetWalukiewicz:2007} generalize the above classes to \emph{call driven automata}.
They show that general \vpa can be exponentially more succinct than the three classes presented.
Additionally, they propose another class called \emph{block \vpa} for which a unique minimum exists that is at most quadratic in the size of some minimal (general) \vpa.
However, to find it, the ``right'' partition into modules must be chosen, for which no efficient algorithm is known.

All above \vpa classes have in common that the languages recognized are subsets of \wellmatched, the states are partitioned into modules, and the minimal automaton (respecting the partition) can be found by quotienting.
While the latter is an enjoyable property from the algorithmic view, the constraints limit practical applicability:
Even under the assumption that the input \vpa recognizes a well-matched language, if it does not meet the constraints, it must first be converted to the respective form.
This conversion generally introduces an exponential blow-up in the number of states.
In contrast, our procedure assumes only weakly-hierarchical \vpa accepting matched-return words.
In general, a weakly-hierarchical \vpa can be obtained with only a linear blow-up.
(In \ultimate{} \automizer the automata already have this form.)

\begin{wrapfigure}[7]{r}{5cm}
	\hspace*{3mm}
	\ifextendedversion{%
		\includegraphics{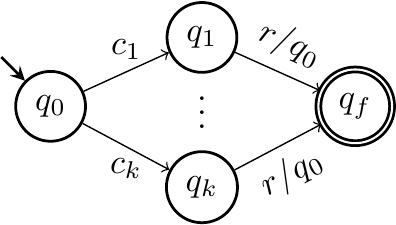}
	}\else{%
		\tikzsetnextfilename{auto_sevpa}
\begin{tikzpicture}
	\node[state] (q0) {$q_0$};
	\node[state] (q1) [right=8mm of q0, yshift=7mm] {$q_1$};
	\node (dots) [below=-1mm of q1] {$\vdots$};
	\node[state] (qk) [below=1mm of dots] {$q_k$};
	\node[state,final] (qf) [right=8mm of q1, yshift=-7mm] {$q_f$};
	\initialtransition[5mm]{q0}
	\draw[->] (q0) to node[above] {$c_1$} (q1);
	\draw[->] (q0) to node[below] {$c_k$} (qk);
	\draw[->,sloped] (q1) to node[above] {$r/q_0$} (qf);
	\draw[->,sloped] (qk) to node[below] {$r/q_0$} (qf);
\end{tikzpicture}
	}\fi

	\caption{A parametric $k$-\sevpa.}
	\label{fig:sevpa}
\end{wrapfigure}

Consider the $k$-\sevpa in Figure~\ref{fig:sevpa}.
It has $k$ modules $\{ q_1 \}$, $\dots$, $\{ q_k \}$ (and the default module $\{ q_0, q_f \}$).
This is the minimal k-\sevpa recognizing the language with the given modules.
Our algorithm will (always) merge all singleton modules into one state, resulting in a (minimal) three-state \vpa.

\smallskip

Caralp \etal~\cite{CaralpEtal:2013} present a polynomial trimming procedure for \vpa.
The task is to ensure that every configuration exhibited in the \vpa is both reachable and co-reachable.
Such a procedure may add new states.
We follow the opposite direction and exploit untrimmed configurations to reduce the number of states.

\smallskip

Ehlers~\cite{Ehlers:2010} provides a \sat encoding of the question ``does there exist an equivalent Büchi automaton (\ba) of size $n-1$''.
%
Baarir and Duret-Lutz~\cite{BaarirDuretlutz:2014,BaarirD15} extend the idea to so-called \emph{transition-based generalized} \ba.
Since the search is global, on the one hand, such a query can be used iteratively to obtain a reduced \ba after each step and some globally minimal \ba upon termination; on the other hand, global search leaves little structure to the solver.

Geldenhuys \etal~\cite{GeldenhuysMZ09} also use a \sat encoding to reduce the state-space of nondeterministic \fa.
The first step is to construct the minimal deterministic \fa \B.
Then the solver symbolically guesses a candidate \fa of a fixed size and checks that the automaton resulting from the subset construction applied to the candidate is isomorphic to \B.
If the formula is unsatisfiable, the candidate size must be increased.
Determinization may incur an exponential blow-up, and the resulting automaton is not always (but often) minimal.

In contrast to the above works, our \pmaxsat encoding consists of constraints about a quotienting relation (which always exists) that is polynomial in the size of the \vpa.
We do not find a minimal \vpa, but our technique can be applied to \vpa of practical relevance (the authors report results for automata with less than $20$ states), in particular using our greedy algorithm.

\smallskip

Restricted to \fa, the definition of a \Qrel coincides with direct bisimulation~\cite{EtessamiEtal:2005,DillEtal:1991}.
This has two consequences.
First, for \fa, we can omit the transitivity clauses because a direct bisimulation is always transitive.
Second, our algorithm always produces the (unique) maximal direct bisimulation.
This can be seen as follows.
If two states $p$ and $q$ bisimulate each other, then \X{p, q} can be assigned \true: since we are looking for a maximal assignment, we will assign this value.
If $p$ and $q$ do not bisimulate each other, then in any satisfying assignment \X{p, q} must be \false.
Alternatively, one can also say that our algorithm searches for \emph{some} maximal fixpoint, which is unique for direct bisimulation.

For \fa, it is well-known that minimization based on direct simulation yields smaller automata compared to direct bisimulation (i.e., the induced equivalence relation is coarser)~\cite{EtessamiEtal:2005}.
Two states can be merged if they simulate each other.
Our \pmaxsat encoding can be generalized to direct simulation by making the variables non-symmetric, i.e., using both $\Xa{p,q}$ and $\Xa{q,p}$ and adapting the clauses in a straightforward way.
This increases the complexity by a polynomial.

\newpage

\iftacasversion{%
\bibliographystyle{abbrv}
}\else{%
\bibliographystyle{abbrv_url}
}\fi
\bibliography{bibl}

\iftacasversion\else{%
\newpage
\appendix

\section{Implementation with an external PMax-SAT solver}\label{app:external_solver}
The most direct implementation of our technique consists of two phases:
1)~Constructing all clauses and 2)~solving the \pmaxsat instance.
We also implemented this approach where we use the established external \pmaxsat solver \texttt{ahmaxsat} in version~$1.68$~\cite{AH15d}.

We compared the implementation that uses our own greedy solver to the one that uses the external solver and finds a globally optimal solution.
For the comparison we created $100$~random \vpa (cf. Appendix~\ref{app:random}) with $50$~states, two internal, call, and return symbols, $50\%$ acceptance density, and $100\%$ transition density.
Return transitions were inserted for each stack symbol.

As we have pointed out in Section~\ref{sec:local_optimum}, an optimal solution to the \pmaxsat instance does not guarantee that the resulting automaton is smaller than for a suboptimal solution (it can even have the opposite effect).
We found that in all cases the resulting automata were identical for both solution methods.
This adds confidence that our greedy solver does not waste reduction potential.

\section{Implementation of an equality solver for checking transitivity incrementally}\label{app:transitivity_checker}
Combining a Boolean solver with a first-order theory solver results in a well-known satisfiability modulo theories (SMT) solver~\cite{KroeningS08}.
To reduce the number of clauses in our Boolean solver, we implemented a theory solver for the equality domain.
The purpose of the solver is 1)~to check that the current partial assignment is consistent and, if it is consistent, 2)~to generate all variables that must be set to \true to not make the current assignment inconsistent.

\smallskip

We first describe the high-level architecture.
Whenever the Boolean solver assigns a variable \X{p, q} the value \true (\false), it asserts equality $p = q$ (disequality $p \neq q$) to the equality solver.
The equality solver checks whether this assertion makes the current context inconsistent.
If the context becomes inconsistent, the equality solver reports a contradiction and the Boolean solver backtracks the assignment.
(Backtracking must also be synchronized with the equality solver.)
If the context stays consistent, the solver returns all additional equalities that follow from the context by transitivity.
The Boolean solver then additionally assigns the respective variables.

\smallskip

On the lower level, the equality solver maintains a union-find data structure (i.e., a forest) of states to keep track of all states that are currently equal.
Assigning to a variable \X{p, q} the value \true corresponds to a union operation of states $p$ and $q$ (i.e., \texttt{union($p$,$q$)}).
Assigning to a variable \X{p, q} the value \false would correspond to checking that states $p$ and $q$ are not connected (i.e., \texttt{find($p$)} $\neq$ \texttt{find($q$)}).
However, we neither need to care for assignments with the value \false, nor need we keep track of disequalities; the reason is that we always feed the Boolean solver with all transitive equality information, and hence the inconsistency detection takes place at the Boolean level.

To support backtracking, the union operation connects two trees only via a temporary node.
These nodes are stored in a list.
The equality solver additionally keeps a stack with a frame for each decision step of the Boolean solver.
For each decision step, the current list is pushed on the stack and a new list is created.
For each backtracking step, all temporary nodes in the current list are removed and the list is updated to the one on top of the stack (after popping).

\section{Benchmark automata generation}\label{app:automata}
The set of automata that we used in our evaluation was obtained as follows.
The software verifier \ultimate{} \automizer was run on C program benchmarks from the \svcomp 2016~\cite{Beyer16}.
All verification tasks that were solved in less than five iterations were ignored.
(Rationale: automata from early iterations are very similar to the control flow graph and hence not amenable to minimization.)
For the remaining verification tasks, the largest automaton that occurred was written to a file.
The available memory was restricted to $4$~GiB and the timeout was set to $60$~s.
Two different properties were checked: reachability of an error location and termination.
For the former property we used C files from the categories \textit{Arrays, BitVectors, IntegersControlFlow, and DeviceDriversLinux64}.
For the latter property we used C files from the same categories and additionally from the category \textit{Termination}.
From each folder we chose at most $50$ (randomly selected) benchmarks.

\section{Random automata generation à la Tabakov-Vardi}\label{app:random}
We generalize the random model for Büchi automata by Tabakov and Vardi~\cite{TabakovV05} to \vpa.
The original model takes as input the number of states $|Q|$, the number of symbols $|\AlphaI|$, the acceptance density $d_a$ (a number between~$0$ and~$1$), and the transition density $d_t$ (a non-negative number).
The density is the number $d$ such that $d = \frac{k}{|Q|}$, where $k$ is the actual number of accepting states (\resp transitions).

For \vpa we add three more parameters: the number of call symbols $|\AlphaC|$, the number of return symbols $|\AlphaR|$, and the stack symbol density $d_s$ (we consider the same transition density for internal, call, and return transitions).
The latter controls for how many stack symbols a return transition is inserted.
For instance, when inserting a return transition from state $p$ to state $q$ with return symbol $r$ and stack symbol density $d_s$, we insert $k$ return transitions $(q, r, \hier{q}, q')$ with random stack symbols \hier{q} such that $d_s = \frac{k}{|Q|}$.
(Recall that the set of states $Q$ is also the stack alphabet for weakly-hierarchical \vpa.)
}\fi

\end{document}